\documentclass{amsart}

\usepackage{graphicx}
\usepackage{tikz}

\newtheorem{theorem}{Theorem}[section]
\newtheorem{lemma}[theorem]{Lemma}

\theoremstyle{definition}
\newtheorem{definition}[theorem]{Definition}
\newtheorem{example}[theorem]{Example}

\theoremstyle{remark}
\newtheorem{remark}[theorem]{Remark}

\numberwithin{equation}{section}

\newcommand{\alg}[1]{\mathfrak{#1}}
\newcommand{\mc}[1]{\mathcal{#1}}
\newcommand{\diag}[0]{\operatorname{diag}}
\newcommand{\Tr}[0]{\operatorname{Tr}}
\newcommand{\supp}[0]{\operatorname{supp}}

\begin{document}
\title{Subfactors and quantum information theory}

\author{Pieter Naaijkens}
\address{Department of Mathematics, University of California, Davis, One Shields Ave, Davis CA 95616, United States; JARA Institute for Quantum information, RWTH Aachen University, Otto-Blumenthal-Stra{\ss}e 20, 52074 Aachen, Germany}
\email{naaijkens@physik.rwth-aachen.de}
\thanks{This project has received funding from the European Union's Horizon 2020 research and innovation programme under the Marie Sklodowska-Curie grant agreement No 657004. I would also like to thank Leander Fiedler for interesting discussions.}


\date{April 3, 2018}

\begin{abstract}
	We consider quantum information tasks in an operator algebraic setting, where we consider normal states on von Neumann algebras. In particular, we consider subfactors $\alg{N} \subset \alg{M}$, that is, unital inclusions of von Neumann algebras with trivial center. One can ask the following question: given a normal state $\omega$ on $\alg{M}$, how much can one learn by only doing measurements from $\alg{N}$? We argue how the Jones index $[\alg{M}:\alg{N}]$ can be used to give a quantitative answer to this, showing how the rich theory of subfactors can be used in a quantum information context. As an example we discuss how the Jones index can be used in the context of wiretap channels.

Subfactors also occur naturally in physics. Here we discuss two examples: rational conformal field theories and Kitaev's toric code on the plane, a prototypical example of a topologically ordered model. There we can directly relate aspects of the general setting to physical properties such as the quantum dimension of the excitations. In the example of the toric code we also show how we can calculate the index via an approximation with finite dimensional systems. This explicit construction sheds more light on the connection between topological order and the Jones index.
\end{abstract}

\maketitle

\section{Introduction}
Quantum information can be defined as the study of information processing in a quantum setting, and is the quantum analog of classical (Shannon) information theory. In a typical scenario, Alice encodes classical information in a quantum state $\rho$, and sends it in some way to Bob. This transmission is rarely perfect, so Bob receives a state $\mathcal{E}(\rho)$, where $\mathcal{E}$ is a completely positive map (a \emph{quantum channel}) describing how the state is affected by the transmission. Bob's task is then to recover the classical information that Alice wanted to send, and a natural question is to ask what the rate of information is that can be send in this way (which depends on the channel $\mathcal{E}$). By now there are multiple textbooks which provide an introduction to the field, see for example~\cite{MR2986302,MR1796805,MR3088659}.

Usually only finite dimensional systems are considered (or infinitely many copies of a finite systems, in case one is interested in asymptotics). Comparatively little work has been done on quantum systems with \emph{infinitely} many degrees of freedom. Nevertheless, there are some examples, such as a study of the Holevo capacity~\cite{MR2200885} and a discussion of error correction in infinite Hilbert spaces~\cite{BenyKempfKribs2009} (this list is by no means complete). See also~\cite{MR1230389} for an overview of some technical properties of quantum channels in the operator algebraic picture.

One of the main goals of this paper is to show that the theory of operator algebras gives us new tools to study quantum information tasks in infinite systems.
It is our hope that the exposition is accessible to both quantum information theorists as well as to operator algebraists. We do this by relating more technical operator algebraic aspects to more familiar tasks in quantum information. Our running example will be subfactors, unital inclusions $\alg{N} \subset \alg{M}$ of von Neumann algebras with trivial center (see the next section for a brief introduction to von Neumann algebras). Subfactors are extensively studied in the operator algebra community, and have found applications in other fields as well (see for example~\cite{MR1838752,MR1027496,klindex}).

Rather than focussing on the states of a system, in our infinite dimensional setting we find it easier and more transparent to work in the Heisenberg picture, with a focus on the observables of the theory. The duality between the state (or Schr{\"o}dinger) picture and the Heisenberg picture is of course well-known in quantum mechanics. So for us a quantum channel will be a normal unital completely positive map $\mc{E} : \alg{M} \to \alg{N}$. A particular example is a \emph{conditional expectation}, in which case $\mathcal{E}$ is a ``projection'' from $\alg{M}$ onto a subalgebra $\alg{N}$. We will give a precise definition in Section~\ref{sec:opalg}. It can be thought of as a generalisation of the partial trace (which indeed is an example of a conditional expectation) that does not require a bipartite structure.

The physical interpretation is as follows. We can think of $\alg{M}$ modelling the observables in some quantum system. The subalgebra $\alg{N}$ then describes a subsystem, in the sense that only a subset of the observables can be measured. In particular, there may be states on $\alg{M}$ that can be distinguished with measurements in $\alg{M}$, but now when one is restricted to $\alg{N}$. We can then think of a conditional expectation $\mathcal{E}: \alg{M} \to \alg{N}$ as describing the quantum operation of restricting observables in $\alg{M}$ to the subsystem, or alternatively describing how much of their ability to distinguish states they lose. In the dual picture, $\mathcal{E}$ can be used to extend states on $\alg{N}$ to the bigger system $\alg{M}$ in a canonical way. In fact, $\omega \circ \mathcal{E}$ is the unique extension of a state $\omega$ on $\alg{N}$ that minimizes the relative entropy with respect to any invariant state $\varphi \circ \mathcal{E} = \varphi$ of $\alg{M}$.

The interesting thing about subfactors is that there typically is a natural or ``preferred'' conditional expectation. To each such conditional expectation one can define an ``index''. There is a unique conditional expectation $\mc{E}$ minimizing this index, and the corresponding index is written as $[\alg{M} : \alg{N} ]$~\cite{MR696688,MR829381,MR1027496}.\footnote{For Type II$_1$ factors one should choose the trace-invariant conditional expectation. In general this is not necessarily the one that minimizes the index.} This is called the \emph{Jones index}. We recall the definition in Section~\ref{sec:subfactors}. The upshot is that this index can be interpreted as an information quantity related to question discussed above.
This is done through expressing the index in terms of Holevo's $\chi$-quantity~\cite{MR0398304}, which gives a bound on the amount of classical information one can recover from an ensemble of states.
Although such normal conditional expectations are important in the theory of subfactors, applications to quantum information have received little attention. Here we try to fill this gap, and show how subfactors give rise to examples of quantum channels, which have a natural operational interpretation. 

We also discuss two examples of physical systems where a subfactor $\alg{N} \subset \alg{M}$ arises naturally: rational conformal field theory~\cite{MR1838752} and the toric code on the plane~\cite{klindex}. In these systems, the index is related to the superselection sectors (or ``charges'') of the model, providing a physical mechanism with which additional states can be distinguished. Some of these ideas have been discussed earlier in~\cite{phdleander,secretjkl}. Here we provide a more mathematical account of the main ingredients.

The definition of the index is rather technical and in many cases not amenable to a direct calculation. In particular for models such as the toric code it would be helpful to have an approximation procedure that gives us the index from (a sequence of) finite dimensional models. Here we present and work out a new way to do precisely that in the toric code. This also clarifies the role of the superselection sectors of the model in the quantum information task described above. It provides a way to relate a method introduced by Haah~\cite{MR3465431} to define the charges in a model, to the definition of superselection sectors that is appropriate for infinite systems.

The paper is organized as follows. We first recall the mathematical setting of quantum information in infinite systems. Then, in Section~\ref{sec:relent}, we recall the definition of relative entropy and discuss how it can be used to say something on how well two states can be distinguished. Section~\ref{sec:subfactors} gives a recap of the main definitions in subfactor theory, and shows how the Jones index is related to the relative entropy. This is applied in Section~\ref{sec:wiretap} to a wiretap scenario. Finally, in the last two sections we give examples of how such subfactors occur in physical systems, and show how in the toric code one can use approximations by finite dimensional systems to calculate the index. 

\section{Quantum channels and operator algebras}\label{sec:opalg}
Often in quantum information theory one considers finite dimensional systems. That is, systems given by (tensor products of) finite dimensional Hilbert spaces. Let $\mc{H}$ be such a finite dimensional Hilbert space. The physical operations on $\mc{H}$ are described by $\alg{B}(\mc{H})$, with self-adjoint elements corresponding to measurable observables, and states are given by density operators on $\mc{H}$ (that is, positive operators on $\mc{H}$ with unit trace). A \emph{quantum channel} $\Phi$ is a linear map that sends states on a system $\mc{H}_A$ to states on a possibly different system $\mc{H}_B$. Mathematically, it can be modeled as a trace-preserving completely positive linear map between the bounded linear operators of each of these systems. Positivity makes sure that positive operators are sent to positive operators, and hence because of the trace-preserving condition, density operators (or states) to density operators. Complete positivity implies that we can consider the composite system $\mc{H}_A \otimes \mathbb{C}^n$, and $\Phi \otimes \operatorname{Id}_n$ will still be a quantum channel.

Physically a quantum channel describes what happens with a state after some action. Examples range from time evolution over a finite amount of time, restricting a state to a smaller set of observables, or describing what happens if we sent a quantum state from one place to another subject to noise induced by interaction with the environment (a ``noisy channel''). Typical questions are then, for example, what the maximum rate of (classical or quantum) information is that can be sent through this channel, and if we can improve this by using error correction schemes. We refer to~\cite{MR3088659} for an in-depth treatment of many applications and results.

Here we will be interested in infinite quantum systems, which are no longer described by finite dimensional Hilbert spaces. In addition, the algebra generated by the observables of the system is no longer $\alg{B}(\mc{H})$, but rather some subalgebra $\alg{M} \subset \alg{B}(\mc{H})$, in particular those subalgebras that are von Neumann algebras.
We will recall some of the basic definitions and facts of the theory of such operator algebras, although some familiarity with the theory of von Neumann algebras will be helpful.
An introduction to von Neumann algebras and their role in physics can be found in~\cite{MR887100}.
See also~\cite{KeylSchlingemannWerner} for a brief discussion in the context of quantum information.

For concreteness, suppose that $\alg{M} \subset \alg{B}(\mathcal{H})$ is a von Neumann algebra.
That is, $\alg{M}$ is closed under the $*$-operation (often called the Hermitean conjugate in physics), and satisfies $\alg{M} = \alg{M}''$.
Here the prime denotes the commutant, $$\alg{M}' = \{ X \in \alg{B}(\mathcal{H}) : [A,X] = 0 \,\forall A \in \alg{M} \}.$$
Hence $\alg{M}$ is equal to its double commutant.
There are different topologies one can consider on a von Neumann algebra. 
It is most natural to consider \emph{normal} maps.
A linear functional $\varphi : \alg{M} \to \mathbb{C}$ is normal if there are sequences $\xi_n, \psi_n \in \mathcal{H}$ with $\sum_n \| \xi_n \|^2 < \infty$ (and similarly for $\psi_n$), such that $\varphi(X) = \sum \langle \xi_n, X \psi_n\rangle$ for all $X \in \alg{M}$.
A linear map $\mathcal{E} : \alg{M} \to \alg{N}$ is normal if and only if $\omega \circ \mathcal{E}$ is a normal state on $\alg{M}$ for any normal state $\omega$ of $\alg{N}$. Equivalently, $\sup_\lambda \mathcal{E}(X_\lambda) = \mathcal{E}(\sup X_\lambda)$ for any increasing and bounded (in norm) net $X_\lambda$ of positive elements in $\alg{M}$.
We will only consider normal maps.

A von Neumann algebra is called a \emph{factor} if $\alg{M} \cap \alg{M}' = \mathbb{C} I$, i.e.\ its center is trivial.
Every von Neumann algebra can be decomposed as a direct sum (or, more generally, a direct integral) of factors.
As was already known to von Neumann, factors can be classified into three types, denoted I, II and III (with further subdivisions possible).
Factors of Type I are the most familiar: those are precisely those von Neumann algebras that are isomorphic to $\alg{B}(\mathcal{H})$ for some Hilbert space $\mathcal{H}$.
They are completely classified by the dimension of the Hilbert space.
For example, a factor of Type I$_n$ is isomorphic to $M_n(\mathbb{C})$.
Every finite dimensional von Neumann algebra is a direct sum of Type $I$ factors.

We will be most concerned with infinite dimensional systems.
In particular, in the language of von Neumann algebras, we are primarily interested in systems where the observables are modeled by a von Neumann algebra that is not of Type I.
This already makes the situation very different from the finite dimensional setting.
Consider for example the case where the observables are given by a Type III factor $\alg{M}$. In such a factor, every projection $P$ is Murray-von Neumann equivalent to $I$, that is, there is an isometry $W \in \alg{M}$ such that $W^*W = I$ and $W W^* = P$. A consequence is that there a \emph{no} normal pure states.\footnote{Since $\alg{M}$ is in particular a $C^*$-algebra, it \emph{does} have pure states, but these states cannot be normal.} Let $\omega$ be a normal state implemented by a vector $\Omega$ (this can always be achieved in a suitable representation). If $\alg{M}$ is Type III, then also $\alg{M}'$ is of Type III, and by choosing a non-trivial projection $P$ in $\alg{M}'$ (which always exists), we can find isometries $V,W$ such that $V V^* + W W^* = I$. Then, with $A \in \alg{M}$, 
\[
	\omega(A) = \langle V^* \Omega, A V^* \Omega \rangle + \langle W^* \Omega, A W^* \Omega \rangle.
\]
By a suitable choice of $V$ and $W$ the states on the right hand side can be made distinct and hence $\omega$ is not a pure state.

This is not just a mathematical curiosity, but also of physical relevance. For example, under quite natural assumptions one can show that the local observables in relativistic quantum field theory are Type III factors. See~\cite{Yngvason2015} for a discussion of some of the consequences this has, for example on the entanglement properties of the vacuum. It is therefore indeed useful to consider the more general operator algebraic setting.

As mentioned an important role is played by quantum channels. In the operator algebraic picture it is more natural to consider the Heisenberg picture. That is, rather than viewing a channel as a map of states, we view it as a map between observable algebras: in the Heisenberg picture a \emph{quantum channel} is a normal unital completely positive map $\mc{E}: \alg{M} \to \alg{N}$ between von Neumann algebras. Of course these two notions of a channel are dual, and there is a unique ``adjoint'' map $\mc{E}_{*} : \alg{N}_* \to \alg{M}_*$, which sends a normal state $\omega$ on $\alg{N}$ to a normal state $\omega \circ \mc{E}$ on $\alg{N}$. Conversely, if $\mc{E}_*$ is completely positive, it uniquely defines a normal unital completely positive map $\mc{E} : \alg{M} \to \alg{N}$.

A fundamental result of Stinespring says that any completely positive map $\mathcal{E} : \alg{M} \to \alg{B}(\mathcal{H})$ is of the form $\mathcal{E}(X) = V^* \pi(X) V$, where $\pi : \alg{M} \to \alg{B}(\mc{K})$ is a representation of $\alg{M}$ on some Hilbert space $\mc{K}$, and bounded linear map $V : \mathcal{H} \to \mathcal{K}$~\cite{MR0069403}. If $\mathcal{E}$ is unital, $V$ is an isometry. If $\mathcal{E}$ is normal, then $\pi$ can also be chosen to be normal. If all Hilbert spaces are finite dimensional, a consequence of Stinespring's theorem is that $\mathcal{E}(X) = \sum_{i=1}^N V_i X V_i^*$, where the $V_i$ are called the \emph{Kraus operators}. For infinite systems this is more subtle, but see for example~\cite{MR2718321}.

A particular example of a quantum channel is a \emph{conditional expectation} $\mc{E} : \alg{M} \to \alg{N}$ from a von Neumann algebra $\alg{M}$ onto a subalgebra $\alg{N}$ (more generally, one could also consider inclusions of $C^*$-algebras).
Conditional expectations play a fundamental role in subfactor theory.
One can find different (but equivalent) definitions in the literature. We will use the following one:
\begin{definition}
	A \emph{conditional expectation} from a von Neumann algebra (or, more generally, a $C^*$-algebra) $\alg{M}$ onto $\alg{N}$ is a positive linear map $\mc{E}: \alg{M} \to \alg{N}$ such that $\mc{E}(ABC) = A \mc{E}(B) C$ and $\mc{E}(A) = A$ for all $A,C \in \alg{N}$ and $B \in \alg{M}$. 
\end{definition}
We will only consider unital conditional expectations between von Neumann algebras, and assume that they are normal. Note that if $\mathcal{E}$ is unital, the condition that $\mc{E}(A) = A$ for $A \in \alg{N}$ already follows from $\mc{E}(ABC) = A \mc{E}(B) C$.

Conditional expectations were introduced by Umegaki~\cite{MR0068751} and Dixmier~\cite{MR0059485} as a non-commutative generalization of conditional expectations in probability theory.
An important result by Tomiyama says that every projection of norm one, that is, a linear map $\mathcal{E}$ from a $C^*$-algebra $\alg{M}$ onto a $C^*$-subalgebra of $\alg{M}$ such that $\mathcal{E} \circ \mathcal{E} = \mathcal{E}$ and $\| \mathcal{E} \| = 1$, is in fact a conditional expectation~\cite{tomiyama1957}.
The converse is also true: every conditional expectation is a projection of norm one.
Finally, it can be shown that $\mc{E}$ is in fact \emph{completely} positive. In particular, $\mc{E}$ is a quantum channel. 
Proofs of all these statements can be found in~\cite[\S 9]{MR696172} or section II.6.10 of~\cite{MR2188261}. 

\begin{example}
	Consider a bipartite system $\mathcal{H} = \mathcal{H}_1 \otimes \mc{H}_2 = \mathbb{C}^m \otimes \mathbb{C}^n$. Let $\tau$ be the normalized trace on $\mathcal{H}_2$. Then we can define a map $\operatorname{Tr}_B$ by demanding $\operatorname{Tr}_2(A \otimes B) = A \tau(B)$ for all $A \in \alg{B}(\mc{H}_1)$ and $B \in \alg{B}(\mc{H}_2)$. By identifying $A$ with $A \otimes I$, this defines a linear map $\operatorname{Tr}_2 : \alg{M} \to \alg{N}$, with $\alg{M} = \alg{B}(\mc{H}_1 \otimes \mc{H}_2$) and $\alg{N} = \alg{B}(\mc{H}_1) \otimes \mc{H}_2$. It is easy to check that $\operatorname{Tr}_2$ is a conditional expectation. Note that up to a normalization factor, this is the partial trace of the second system.
\end{example}

The existence of a conditional expectation is not guaranteed and implies some conditions on $\alg{N}$. For example, if $\mc{E} : \alg{M} \to \alg{N}$ is a normal conditional expectation, and $\alg{M}$ is of Type I, then so is $\alg{N}$. Similarly, if $\alg{M}$ is semi-finite, so is $\alg{N}$ (cf.~\cite[\S 10]{MR696172}).

\section{Distinguishing states and relative entropy}\label{sec:relent}
Entropies and relative entropies play an essential role in quantum (and indeed, also classical) information theory. Many tasks, such as determining the amount of information one can send through a channel, ultimately boil down to calculation of certain entropies. Another such task that will be relevant for us is distinguishing states. Before we recall how entropies play a role here, we recall the definition of the relative entropy in the context of von Neumann algebras. This was first introduced by Araki~\cite{MR0425631,MR0454656}. A more modern treatment as well as an overview of subsequent results can be found in~\cite{MR1230389}.

Let $\alg{M}$ be a von Neumann algebra and suppose that $\omega,\varphi$ are two positive normal functionals on $\alg{M}$. Moreover, suppose that $\omega$ is implemented by a vector $\xi$, which can always be realized by switching to the Haagerup standard form if necessary~\cite{MR0407615}. The vector $\xi$ induces a positive functional $\omega_{\xi}'$ on the commutant $\alg{M}'$, and it is possible to define the \emph{spatial derivative} $\Delta(\varphi/\omega'_\xi)$~\cite{MR561983}, which generally is an unbounded operator. The \emph{relative entropy} is then defined as\footnote{The reader should be warned that in the literature sometimes the order of the arguments is reversed in the definition.}
\begin{equation}
	S(\omega,\varphi) := \begin{cases}
		- \langle \xi, \log \Delta(\varphi/\omega'_\xi) \xi\rangle & \textrm{if } \supp \omega \leq \supp \varphi \\
		+ \infty & \textrm{otherwise}
			\end{cases}.
\end{equation}
Here $\supp \omega$ is the support projection of $\omega$, that is, the smallest projection $P$ such that $\omega(P) = \omega(I)$. Note that we do not restrict to states, but consider \emph{all} positive normal linear functionals.

For the remainder of this paper the precise technical details of this definition are not important. We will however frequently use that for finite dimensional systems, the definition reduces to the following equation, which will be more familiar to the quantum information community:
\begin{equation}
S(\rho,\sigma) = \Tr( \rho \log(\rho) - \rho \log(\sigma))
\end{equation}
if $\supp(\rho) \leq \supp(\sigma)$, and $+ \infty$ otherwise. Here we identify the positive linear functionals $\rho$ and $\sigma$ with the corresponding positive matrices, i.e. $\rho(A) = \Tr(\rho A)$. This notion of relative entropy was first studied by Umegaki~\cite{MR0142006}.

A fundamental result is that for normal states $S(\omega,\varphi) \geq 0$, with equality if and only if $\omega = \varphi$. So even though $S(\omega,\varphi)$ is not a metric, it tells us something on how distinct the two states are. It can sometimes tell us more than just being distinct. Suppose that we have a normal state $\varphi$. Then we can consider a finite probability distribution $\{ p_x \}$ and a corresponding set of normal states $\{ \varphi_x \}$ such that $\varphi = \sum_x p_x \varphi_x$. Physically this corresponds to a procedure where we prepare a state by choosing one of the states $\varphi_x$ with probability $p_x$. Suppose now that after preparing the state we give it to Bob, who is allowed to know from which ensemble of states the state is selected, but does not know the probability distribution $\{p_x\}$. Bob's task is to recover the probability distribution. If the states $\varphi_x$ have orthogonal support, Bob can recover the distribution arbitrarily well, given enough copies of the state $\varphi$.

This is no longer true if the states do not have orthogonal support. Nevertheless, it is still possible to give a bound on the amount of information that can be obtained. To this end, define the \emph{Holevo $\chi$-quantity} for such a decomposition as~\cite{MR0398304}
\begin{equation}
	\label{eq:chi}
	\chi(\{p_x\}, \{\varphi_x\}) := \sum_{x} p_x S(\varphi_x, \varphi).
\end{equation}
It can also be written in terms of entropies (rather than \emph{relative} entropies), but for infinite systems the relative entropy formulation is preferred (since entropies are often infinite in that case). The $\chi$-quantity gives a bound on the amount of classical information that can be recovered from an ensemble, and appears in many different places in quantum information. For example, a fundamental result is the \emph{coding theorem}, which says that $\max_{\{p_x, \varphi_x\}} \chi(\{p_x\}, \{\varphi_x\})$ is equal to the \emph{classical} capacity of the classical-quantum channel $x \mapsto \varphi_x$~\cite{MR1486663,PhysRevA.56.131}. That is, the rate of classical information that can be transmitted without error in the limit of asymptotically many uses of the channel.

For infinite systems, however, the quantity~\eqref{eq:chi} is not that useful in itself: for faithful normal states $\varphi$ of von Neumann algebras of Type II or Type III, we can make this quantity as large as we want by choosing an appropriate decomposition of $\varphi$ (this follows from~\cite[Lemma 6.10]{MR1230389}). Hence while for a given decomposition this is still a meaningful quantity, one has to take care when optimizing over all possible decompositions.

Here we will consider a slightly different scenario. Suppose we have a unital inclusion of von Neumann algebras $\alg{N} \subset \alg{M}$. We will think of $\alg{M}$ as describing ``full'' set of observables, while $\alg{N}$ is a more limited set of observables, for example describing the situation where a third party has only access to a limited part of the whole quantum system. Then if $\omega$ is a normal state on $\alg{M}$, we can again consider a decomposition $\omega = \sum p_x \omega_x$ as before. Since $\alg{M}$ contains more operations than $\alg{N}$, one would expect that in general an observer who has access to all operations in $\alg{M}$ would do a better job of recovering $\{p_x\}$ than an observer who only has access to operations in $\alg{N}$. That is, an observer who has to work with the \emph{restricted} states $\omega_x | \alg{N}$. From the discussion above it is reasonable to stipulate that the advantage that $\alg{M}$ has over $\alg{N}$ can be quantified by
\begin{equation}
	\label{eq:chidiff}
	S_{\alg{M}|\alg{N}}(\{p_x\}, \{\varphi_x\}) := \chi( \{p_x \}, \{\varphi_x\}) - \chi( \{p_x\}, \{ \varphi_x | \mathfrak{N} \}).
\end{equation}
This quantity is also called the \emph{entropic disturbance} of the quantum channel that restricts normal states on $\alg{M}$ to $\alg{N}$ and will play a central role in the remainder of this article. Some properties and applications of this quantity in infinite dimensional systems are discussed in~\cite{2016arXiv160802203S}. For other applications, see Section~\ref{sec:wiretap} below.

\section{Subfactors, conditional expectations, and Jones index}\label{sec:subfactors}
We now restrict to a particularly interesting class of examples of inclusions $\alg{N} \subset \alg{M}$: the subfactors. There are many aspects to such subfactors, but for our purposes one of the most important features is that one can define an \emph{index} $[\alg{M} : \alg{N}]$ which, heuristically speaking, gives a measure of how much bigger $\alg{M}$ is compared to $\alg{N}$. It can be seen as the generalization of the index of a subgroup in a group. Subfactors and the index were first studied by Jones in the Type II$_1$ case~\cite{MR696688}. Many of the results were later extended to general factors by Kosaki~\cite{MR829381} and Longo~\cite{MR1027496,MR1059320}, among others.

We first recall the definition of a subfactor.
\begin{definition}
A \emph{subfactor} is a unital inclusion of factorial von Neumann algebras $\alg{N} \subset \alg{M}$. It is called \emph{irreducible} if $\alg{N}' \cap \alg{M} = \mathbb{C} I$. 
\end{definition}
The theory of subfactors is very rich and examples come up in many different contexts, some of which we will discuss in Section~\ref{sec:examples}.

\begin{example}
	Let $\alg{M} = \alg{B}(\mc{H}_1 \otimes \mc{H}_2)$ and $\alg{N} = \alg{B}(\mc{H}_1) \otimes I$. Then $\alg{N} \subset \alg{M}$ is a subfactor. It is not irreducible (unless $\dim \mc{H}_2 = 1$), since $\alg{N}' \cap \alg{M} = (\alg{B}(\mc{H}_1) \otimes I)' = I \otimes \alg{B}(\mc{H}_2)$.
\end{example}

A key role in the theory is played by the conditional expectations introduced in Section~\ref{sec:opalg}.
We will mainly consider irreducible subfactors. In that case, if there is a faithful conditional expectation $\mathcal{E} : \alg{M} \to \alg{N}$, it is unique~\cite[Sect. 5]{MR1027496}.
We are now in a position to define the \emph{index} of a subfactor $\alg{N} \subset \alg{M}$. This is the quantity that will allow us to connect an inclusion of von Neumann algebras to a quantum information quantity.
\begin{definition}\label{def:index}
Let $\alg{N} \subset \alg{M}$ be an irreducible subfactor with both algebras not of Type I. Let $\mc{E}: \alg{M} \to \alg{N}$ be the unique faithful conditional expectation (if it exists). Then we define the Jones (or Jones-Kosaki-Longo) index by
\[
	[\alg{M}:\alg{N}] := (\sup \{ \lambda > 0 : \mc{E}(x) \geq \lambda x \textrm{ for all } x \in \alg{M}_+ \})^{-1}.
\]
If such a conditional expectation does not exist, we set $[\alg{M}:\alg{N}] = \infty$.
\end{definition}
The index measures how much bigger $\alg{M}$ is compared to $\alg{N}$. One can show that $[\alg{M}: \alg{N}] \geq 1$, with equality if and only if $\alg{M} = \alg{N}$. It also gives a bound on the dimension of $\alg{M}$ seen as a module over $\alg{N}$ (see for example~\cite[Sect. 3.4]{MR1662525}).

This is not the original definition of Jones, but coincides with his if $\mc{E}$ is the trace-preserving conditional expectation of a Type II$_1$ subfactor. The equivalence to the definition given here is due to Pimsner and Popa in the Type II$_1$ case~\cite{MR860811}. For general subactors, one can define $\operatorname{Ind}(\mc{E})$ for a conditional expectation $\mc{E} : \alg{M} \to \alg{N}$, either through a Pimsner-Popa type inequality or using modular theory (the definition there is quite technical, so we will not repeat it here). It can be shown that there is a unique conditional expectation $\mc{E}_0$ minimizing the index, and one can define $[\alg{M}:\alg{N}]$ using this expectation. It should be noted, however, that for Type II$_1$ factors this minimal conditional expectation need not coincide with the trace-preserving conditional expectation that Jones uses for his index. Since for irreducible subfactors the faithful conditional expectation is unique (if it exists), this is not an issue. We refer to~\cite{MR1662525} for details on the index theory.

A simple argument relates the index to the entropic quantity $S_{\alg{M}|\alg{N}}$ (c.f.~\cite[Cor. 4.1]{MR860811} or~\cite[Prop. 10.2.2]{MR2251116}\footnote{The author would like to thank Ben Hayes for pointing out this reference.}).
\begin{lemma}
	Let $\alg{N} \subset \alg{M}$ be a finite index irreducible subfactor and $\mathcal{E}$ the corresponding conditional expectation. Then for any normal state $\varphi = \varphi \circ \mathcal{E}$ of $\alg{M}$ we have 
	\[
		S_{\alg{M}|\alg{N}}( \{ p_x \}, \{ \varphi_x \}) \leq \log \left[ \alg{M} : \alg{N} \right],
	\]
	where $\varphi = \sum_x p_x \varphi_x$.
\end{lemma}
\begin{proof}
	Since $\mathcal{E}$ is a faithful conditional expectation, by Theorem 5.15 of~\cite{MR1230389} $S(\omega, \varphi \circ \mathcal{E}) = S(\omega | \alg{N}, \varphi | \alg{N}) + S(\omega, \omega \circ \mathcal{E})$ for any normal state $\omega$ on $\alg{M}$.
	Using this we can rewrite equation~\eqref{eq:chidiff} as
\[
	S_{\alg{M}|\alg{N}}( \{ p_x \}, \{ \varphi_x \}) = \sum_x p_x S\left(\varphi_x, \varphi_x \circ \mc{E}\right).
\]
From Definition~\ref{def:index} it follows that $\omega \circ \mathcal{E} \geq [\alg{M}:\alg{N}]^{-1} \omega$ for any normal state $\omega$. This implies that
\[
	\sum_x p_x S(\varphi_x, \varphi_x \circ \mathcal{E}) \leq \sum_x p_x S(\varphi_x, [\alg{M}:\alg{N}]^{-1} \varphi_x) = \sum p_x \log [\alg{M}: \alg{N}],
\]
where in the first step we used Corollary~5.12 of~\cite{MR1230389}, and in the second step the scaling properties of the relative entropy.
\end{proof}

The following key result, due to Pimsner and Popa for Type~II$_1$ factors and Hiai for the general case show that the bound can actually be attained. It is in fact sufficient to consider only \emph{faithful} states $\varphi \circ \mathcal{E} = \varphi$, but since for the applications we are interested in this restriction is not very natural, we consider the general case.
\begin{theorem}[Pimsner-Popa~\cite{MR860811}, Hiai~\cite{MR1150623,MR1096438}]\label{thm:pph}
	Let $\alg{N} \subset \alg{M}$ be an irreducible subfactor with finite index. Then we have
	\[
		\log [\alg{M}:\alg{N}] = \sup_{\varphi : \varphi \circ \mc{E} = \varphi} \sup_{ \{p_x\}, \{\varphi_x\}} S_{\alg{M}|\alg{N}}(\{p_x\}, \{\varphi_x\}).
	\]
	Here $\mc{E} : \alg{M} \to \alg{N}$ is the conditional expectation associated to the subfactor. The first supremum is over all normal states on $\alg{M}$ that leave $\mc{E}$ invariant, and the second is over all finite decompositions of such a $\varphi$.
\end{theorem}
The condition that the subfactor is irreducible is not essential, but simplifies the statement (and is enough for our purposes). The reason is that in this case there is a unique conditional expectation $\mc{E} : \alg{M} \to \alg{N}$. In the general case, one has to consider the \emph{minimal} conditional expectation.

The use of entropies in subfactor theory was motivated by work of Connes and St{\o}rmer, who were interested in a non-commutative generalization of entropies as they are used in classical dynamical systems (see~\cite{MR2251116} for a review). Their relative entropy of two operator algebras provides the connection between the algebraic definition of the index in Definition~\ref{def:index} to the formula in terms of relative entropies of states given in the theorem above. Our motivation is different: we are mainly interested in the application of these results to quantum information theory, in particular to channel capacities, as discussed in the next section. This concludes our brief discussion of subfactors.

\section{Wiretap channels and private classical capacity}\label{sec:wiretap}
There are various quantum information tasks that can be described in the subfactor setting. For example, an application to secret sharing was discussed in~\cite{secretjkl}. Here we outline how we can interpret the subfactor in the context of wiretapping channels. This provides a new application of entropies in subfactors, and the first in quantum information. Classical wiretap channels were introduced by Wyner~\cite{MR0408979}. Their quantum counterparts were first studied by Schumacher and Westmoreland~\cite{PhysRevLett.80.5695}.

Let us first consider, following~\cite[Sect. 10.4]{MR2986302}, the typical setup for a quantum wiretap channel. Consider finite dimensional Hilbert spaces $\mc{H}_A$ (controlled by Alice), $\mc{H_B}$ (controlled by Bob) and $\mc{H}_E$ (controlled by an eavesdropper Eve). Alice uses an isometry $V: \mc{H}_A \to \mc{H}_B \otimes \mc{H}_E$ to encode states $\rho_A$ in the target Hilbert space consisting of Bob's and Eve's parts. Since Bob and Eve can only control their respective parts, they will have access to the states $\rho_B := \Tr_E(V \rho_A V^*)$ and $\rho_E := \Tr_B(V \rho_A V^*)$, respectively. Note that we can describe this by two quantum channels $\Phi_B$ and $\Phi_E$.

Bob can obtain some information on the initial state $\rho_A$ that Alice sent through the channel by doing measurements on his part of the system. But since Eve \emph{also} has access to part of the system, she can also gain some knowledge. A natural question is to ask how much information Alice can send to Bob in such a way that Eve cannot learn anything about the message Alice wants to send to Bob by measurements on her part of the system. From the discussion in Section~\ref{sec:relent} this can be bounded by
\[
	\chi(\{p_x\}, \{\Phi_B(\rho_A^x)\}) - \chi(\{p_x\}, \{\Phi_E(\rho_A^x)\}),
\]
where $\rho_A = \sum_x p_x \rho_A^x$.
This quantity is sometimes called the \emph{quantum privacy}. If each $\rho_A^x$ is a \emph{pure} state, it can be shown to be equal to a quantity called the \emph{coherent information}~\cite{PhysRevLett.80.5695}.

We now come back to our subfactor setting, with $\alg{N} \subset \alg{M}$ a subfactor, and want to relate it to the wiretapping scenario above. In the Heisenberg picture we can model Bob and Eve with two observable algebras that mutually commute. However, in the general von Neumann algebra setting, there need not be a corresponding tensor product decomposition as above, so it is not clear how to define the channels $\Phi_B$ and $\Phi_E$ as before. Nevertheless, we can consider a scenario that is similar: we suppose that Eve has control over all observables in $\alg{N}$. Bob, on the other hand, is more powerful, and can access $\alg{M}$. Again Alice wants to encode information for Bob that Eve cannot recover. We consider the simplest scenario possible, where Alice encodes the information directly in normal states $\omega$ on $\alg{M}$ and does not send them through a channel first. Then a bound on the amount that Bob can hide from Eve is given by equation~\eqref{eq:chidiff}. Note that the restriction to $\alg{N}$ is a quantum channel: it is the adjoint of the channel (in the Heisenberg picture) $\iota : \alg{N} \to \alg{M}$, the inclusion homomorphism. This channel plays the role of $\Phi_E$. But using subfactor theory we can say more: if Alice is restricted to \emph{ensembles} $\omega = \omega \circ \mc{E}$, the optimum value is given by the (log of the) index $[\alg{M} : \alg{N}]$ by Theorem~\ref{thm:pph}! Again, it is \emph{not} necessary that the individual states in the ensemble are also invariant with respect to $\mc{E}$. 

Let us also briefly comment on the condition that the ensembles are invariant with respect to $\mc{E}$. Let $\omega \circ \mc{E} = \omega$ be a normal state. If $\varphi_\alg{N}$ is a normal state on $\alg{N}$, we can extend it to a normal state on $\alg{M}$ by $\varphi_\alg{N} \circ \mc{E}$. But there can be (and in general, are) many extensions of $\varphi_\alg{N}$ to $\alg{M}$. Then $\varphi_\alg{N} \circ \mc{E}$ is the \emph{unique} extension $\varphi$ that minimizes $S(\omega, \varphi)$ (under the condition that $S(\omega,\varphi) < \infty$ for at least one extension), see the discussion after~\cite[Lemma 5.18]{MR1230389}. So be restricting to such states, we isolate the contribution that is due to Bob being able to distinguish individual states in the ensemble, that cannot be distinguished from the ensemble itself.

What if we consider $n$ copies of the system, can we do better? Alice can send states $\rho^{(n)} = \sum_x p_x^{(n)} \rho_x^{(n)}$, with $\rho_x^{(n)}$ states on $\mc{H}_A^{\otimes n}$ for some integer $n$. The states that Bob receive are described by the channels $\Phi_B \otimes \cdots \otimes \Phi_B$, and similarly for Eve. Note that the states $\rho^{(n)}$ need not be product states (or even separable states), and neither do the allowed measurements have to be of product form. The \emph{private capacity} is then defined as the (average) amount of information that Alice can send to Bob using $n$ channels, in the limit $n \to \infty$, such that the amount of information Eve can recover is negligible. It turns out that this is in fact \emph{equal} to
\[
	\lim_{n \to \infty} \frac{1}{n} \max_{\{p_x^{(n)}, \rho_x^{(n)} \}} \left[ \chi(\{p_x^{(n)}\}, \{\Phi_B^{\otimes n}(\rho_x^{(n)})\})) - \chi(\{p_x^{(n)}\}, \{\Phi_E^{\otimes n}(\rho_x^{(n)})\}) \right],
\]
a result proven by Devetak~\cite{MR2234571} and Cai, Winter and Young~\cite{MR2105852}. Note that this can also be interpreted as a coding theorem.

In the operator algebraic setting we can consider the von Neumann algebra $\alg{M}_1 \overline{\otimes} \alg{M}_2$ generated by the algebraic tensor product $\alg{M}_1 \odot \alg{M}_2$. If $\Phi_i : \alg{M}_i \to \alg {N}_i$, $i=1,2$ are normal maps, there is a normal map $\Phi_1 \otimes \Phi_2 : \alg{M}_1 \overline{\otimes} \alg{M_2} \to \alg{N}_1 \overline{\otimes} \alg{N}_2$ defined in the natural way. Hence we can ask the same question as before, but not with $\alg{M}^{\otimes n}$ and $\alg{N}^{\otimes n}$, respectively: 

\begin{theorem}
Let $\alg{N} \subset \alg{M}$ be an irreducible subfactor. Then for each $n \in \mathbb{N}$ we have
\[
	\sup_{\varphi : \varphi \circ \mc{E}^{\otimes n} = \varphi} \sup_{ \{p_x^{(n)}, \varphi_{x}^{(n)}\}}   
	\left[\chi( \{p_x^{(n)} \}, \{\varphi_x^{(n)}\}) - \chi( \{p_x^{(n)}\}, \{ \varphi_x^{(n)} | \mathfrak{N}^{\otimes n} \}) \right]
	= n \log [\alg{M}:\alg{N}].
\]
A similar statement is true for reducible subfactors if one takes $\mc{E}$ to be the conditional expectation minimizing the index.
\end{theorem}
\begin{proof}
	First note that $\alg{N}^{\otimes n} \subset \alg{M}^{\otimes n}$ is again a subfactor. Hence we can apply Theorem~\ref{thm:pph} with $\mc{E}^{(n)} : \alg{M}^{\otimes n} \to \alg{N}^{\otimes n}$ the minimal conditional expectation for this subfactor. But it turns out that $\mc{E}^{(n)} = \mc{E}^{\otimes n}$, see the proof of Corollary 5.6 of~\cite{MR1027496}. The claim then follows since $[\alg{M}^{\otimes n}: \alg{N}^{\otimes n} ] = [ \alg{M} : \alg{N} ]^n$, by the same corollary.
\end{proof}
This shows that we do not gain any advantage from allowing the use of multiple copies of the system simultaneously, not even when we allow states which are entangled along these copies.
 
In general the private capacity is not additive, in the sense that the capacity of $\Phi \otimes \Phi$ is not twice the capacity of $\Phi$~\cite{PhysRevLett.103.120501}. Here on the other hand it appears that it \emph{is} additive, and it is interesting to see that we can use the special structure of subfactors to explicitly calculate the quantum privacy. It should be noted however that this does not quite prove the private classical capacity of the channel, since it would require an analog of the results in~\cite{MR2105852,MR2234571} (in particular on how to encode information) that is valid in this operator algebraic setting. It nevertheless suggests that using subfactor theory can be very powerful to study certain quantum channels, which warrants further study.

\section{Examples: abelian quantum double models and CFT}\label{sec:examples}
Examples of subfactors arise naturally in physical systems. Here we recall how they appear in the study of superselection sectors in local quantum physics. A more detailed introduction to sector theory and local quantum physics can be found in~\cite{MR1405610}. The goal of the theory is to describe all properties of the charges (or superselection sectors). We will consider two examples here: rational conformal field theory on the circle~\cite{MR1838752}, and a class of topologically ordered quantum spin systems, known as \emph{Kitaev's quantum double}~\cite{MR1951039}. The connection between the latter and the Jones index, together with some applications, was discussed before in~\cite{secretjkl}. 

The interesting thing is that there is a direct physical interpretation of the index in these models. In particular, in both cases the different superselection sectors can be described by a \emph{braided tensor $C^*$-category}. This category encodes all physical properties of the superselection sectors (or charges, anyons), such as fusion and braiding rules. To each such anyon $i$ one can assign a \emph{quantum dimension} (or just dimension) $d_i$. The total quantum dimension, an invariant of the tensor category, is then defined via $\mathcal{D}^2 = \sum_i d_i^2$, where the sum is over a set of representatives of the anyons (or in the language of category theory, of irreducible objects). It plays an important role in topologically ordered systems, where it is believed to be related to an area law for the ground states of such systems~\cite{PhysRevLett.96.110404,PhysRevLett.96.110405}.

The starting point is a net of local algebras, representing the observables that can be measured in finite parts of the system of interest. More precisely, let $\Gamma$ be the space on which the model is defined (for example Minkowksi space for relativistic theories or a discrete lattice for spin systems). Then to suitable subsets $\Lambda \subset \Gamma$ we assign a $C^*$-algebra $\alg{A}(\Lambda)$, representing all observables in the region $\Lambda$. There are two natural properties that we require, the first being isotony: if $\Lambda_1 \subset \Lambda_2$, there should be a unital inclusion $\alg{A}(\Lambda_1) \subset \alg{A}(\Lambda_2)$. The second is locality: if $\Lambda_1 \perp \Lambda_2$ then $\alg{A}(\Lambda_1)$ and $\alg{A}(\Lambda_2)$ commute. Here the meaning of $\Lambda_1 \perp \Lambda_2$ depends on the context: for our purposes it is enough for this to mean that $\Lambda_1$ and $\Lambda_2$ are disjoint. In relativistic theories we require that they are spacelike separated. We also introduce the notation $\Lambda^c := \Gamma \setminus \Lambda$ (again, in relativistic theories one would take the spacelike complement). The assignment $\Lambda \mapsto \alg{A}(\Lambda)$ defines a local net of $C^*$-algebras. The \emph{quasi-local algebra} $\alg{A}$ is the inductive limit of this net (in the category of $C^*$-algebras). It can be interpreted as the set of all observables that can be approximated arbitrarily well by strictly local observables. We also require that $\alg{A}$ is represented on some Hilbert space $\mc{H}$ using a representation $\pi_0$ (typically the vacuum representation).

The type of charges that we want to describe are localized, in the sense that outside of the localization region the system looks like the vacuum or ground state. Hence we need a set $\mathcal{C}$ of localization regions. The elements of $\mathcal{C}$ are subsets of the space $\Gamma$. The choice of admissible regions depends on the type of charges that we want to describe. For example, using Gauss' law we can always detect an electrical charge from arbitrarily far away by measuring the flux through a sphere. Hence such charges cannot be localized in a compact region. Once the set $\mathcal{C}$ is fixed, we consider irreducible representations $\pi$ that look like $\pi_0$ outside the localization region. More precisely, we demand that
\begin{equation}
	\label{eq:sselect}
	\pi | \alg{A}(C^c) \cong \pi_0 | \alg{A}(C^c)
\end{equation}
for \emph{all} $C \in \mathcal{C}$. Here $\cong$ means unitary equivalence of representations. Note that they only have to be unitary equivalent when restricted to observables \emph{outside} of $C$. Nevertheless, this is quite a strong condition, since it has to hold for all regions $C$. Physically this means that we are able to move charges from one region $C_1 \in \mc{C}$ to another region $C_2$. A \emph{sector} is then an equivalence class of irreducible representations satisfying equation~\eqref{eq:sselect}.

The set of such representations has a surprisingly rich structure, which can be analyzed with an additional technical assumption: Haag duality. This says that for each $C \in \mc{C}$ we have $\pi_0(\alg{A}(C))'' = \pi_0(\alg{A}(C^c))'$. One inclusion follows from locality, but the other is non-trivial. The main utility of Haag duality is that it allows one to pass from representations to endomorphisms of $\alg{A}$. That is, each representation $\pi$ as above is equivalent to $\pi_0 \circ \rho$, where $\rho$ is an endomorphism of $\alg{A}$. Moreover, it follows that $\rho(A) = A$ for all $A \in \alg{A}(C^c)$ for some $C \in \mc{C}$. That is, it only acts non-trivially in the localization region $C$. These endomorphisms can be endowed with the structure of a \emph{braided tensor category}, which encodes all physical properties of the charges. For example, it is possible to define a braiding, which tells us what happens if we interchange to charges. For bosons and fermions, interchanging two particles \emph{twice} is always trivial, but in lower dimensional systems there are other possibilities. This happens in fact in both of the examples discussed below.

The final piece of structure that we need is a set of \emph{double localization regions} $\mc{C}_2$, which consists of pairs $C_A, C_B$ of regions in $\mc{C}$, with the condition that $C_A$ and $C_B$ are sufficiently far apart (the precise notion of which is model dependent). We will also write $AB \in \mc{C}_2$. With such a choice we can define the von Neumann algebra $\mc{R}_{AB} := \pi_0(\alg{A}(C_A \cup C_B))''$, that is, the von Neumann algebra generated by $\pi_0(\alg{A}(C_A))''$ and $\pi_0(\alg{A}(C_B))''$. Our final assumption is that the \emph{split property} holds. This says that $\mc{R}_{AB} \simeq \pi_0(\alg{A}(C_A))'' \overline{\otimes} \pi_0(\alg{A}(C_B))''$. Here $\simeq$ means that the two algebras are naturally isomorphic, in the sense that the map $A \otimes B \mapsto AB$ extends to an isomorphism of von Neumann algebras.\footnote{This is not automatic if the algebras are not Type I.} Finally, we can define an algebra $\widehat{\mc{R}}_{AB} := \pi_0(\alg{A}( (C_A \cup C_B)^c))'$. Note that by locality we have $\mc{R}_{AB} \subset \widehat{\mc{R}}_{AB}$. 

It will become clear in a moment why this is an interesting inclusion of algebras to study, but first we give two examples of physical theories that satisfy all the conditions.

\begin{example}[Rational conformal field theory~\cite{MR1838752}]\label{ex:rcft}
	Consider $\Gamma = S^1$, the circle. To each interval we associate a local algebra $\alg{A}(I)$ in such a way that if $I \subset J$, then $\alg{A}(I) \subset \alg{A}(J)$ and all algebras are represented on the same Hilbert space $\mc{H}$. In addition we assume that the net is local, in the sense that $\alg{A}(I)$ and $\alg{A}(J)$ commute if $\overline{I} \cap \overline{J} = \emptyset$ and that the net is irreducible, in the sense that the algebra generated by all $\alg{A}(I)$ is equal to $\alg{B}(\mc{H})$. The representation $\pi_0$ is the identity representation and will be omitted. We also assume that the net is conformal, in the sense that it is covariant with respect to a positive energy representation of the M{\"o}bius group and there is a cyclic vacuum vector. In that case, Haag duality holds for intervals, and so does the split property under a mild additional assumption. Hence all prerequisites to analyze the sectors of the theory are fulfilled.

	The localization regions are the open intervals on the circle. Let $I_A$ and $I_B$ be two intervals with disjoint closure. Then we get a subfactor $\mc{R}_{AB} = \alg{A}(I_A \cup I_B)'' \subset \widehat{\mc{R}}_{AB} = \alg{A}( (I_A \cup I_B)^c)'$. If this subfactor has finite index, we call the conformal field theory \emph{rational}. One can show that the index does not depend on the choice of intervals.
\end{example}

\begin{example}[Kitaev's quantum double model for abelian groups in the thermodynamic limit~\cite{toricendo}]\label{ex:topo}
	Consider the space $\Gamma$ to be the edges of a square $\mathbb{Z}^2$ lattice and let $G$ be a finite abelian group. At each edge we have a $|G|$-dimensional system. The local algebras are given by $\alg{A}(\Lambda) := \bigotimes_{x \in \Lambda} M_{|G|}(\mathbb{C})$ with $\Lambda \subset \Gamma$ a finite subset. This gives a local net as before.

	Dynamics on the system can be introduced by defining local Hamiltonians. For the quantum double model these were first introduced in~\cite{MR1951039}. Once these are defined it is possible to talk about ground states. The model has many ground states, but only one of them is translation invariant~\cite{toricgs}, which we will denote by $\omega_0$. The corresponding GNS representation will be denoted by $\pi_0$. This will be taken as the reference representation. As localization regions $\mc{C}$ we consider \emph{cones}. They can be obtained by taking a point in the plane and draw two straight lines from this point going to infinity. Then we identify $C \subset \Gamma$ with all edges that are contained in the region bounded by the two lines (corresponding to the smallest angle between them), or intersect the two lines. See Figure~\ref{fig:leftcone} for an example of two cones. For the double localization regions $\mathcal{C}_2$ we consider all pairs of disjoint cones such that their distance is at least two. With this choice of localization regions Haag duality and the split property hold~\cite{phdleander,haagdouble,haagdtoric}.
\end{example}

We come back to the inclusions $\mc{R}_{AB} \subset \widehat{\mc{R}}_{AB}$. As mentioned before, a representation $\pi$ satisfying the criterion~\eqref{eq:sselect} can be represented by a localized endomorphism. That is, we can find an endomorphism $\rho_1$ of $\alg{A}$ such that $\rho_1(A) = A$ for all $A \in \alg{A}(C_A^c)$ and $\pi_0 \circ \rho_1 \cong \pi$. But using again~\eqref{eq:sselect}, the same is true for some $\rho_2$ localized in $C_B$, and there is a unitary $V$ such that $V \pi_0 \circ\rho_1(A) = \pi_0 \circ \rho_2(A) V$. Such a $V$ is called a \emph{charge transporter}, since it moves a charge from $C_A$ to $C_B$. Moreover, using locality it follows that $V \in \widehat{\mc{R}}_{AB}$. However, unless $\rho_1$ is trivial, $V$ is \emph{not} in $\mc{R}_{AB}$. Hence it is reasonable to conjecture that $\widehat{\mc{R}}_{AB}$ is bigger than $\mc{R}_{AB}$ precisely because of the charge transporters, and that we can learn something about the charges if we can see how much ``bigger'' it is.

This is the motivation behind the following definition.
\begin{definition}
	Consider a local theory as described above, together with a suitable choice of double localization regions $\mathcal{C}_2$ such that $\mc{R}_{AB} \subset \widehat{\mc{R}}_{AB}$ is a subfactor for every choice of two localization regions $AB \in \mc{C}_2$. The we define $\mu_{\pi_0} := \inf_{AB \in \mc{C}_2} [ \widehat{\mc{R}}_{AB} : \mc{R}_{AB} ]$.
\end{definition}
Typically the index does not depend on the choice of the two localization regions, and one can forget about the infimum.

With this notation we can make the connection between the inclusion $\mc{R}_{AB} \subset \widehat{\mc{R}}_{AB}$ and the sectors of the theory.  
\begin{theorem}[\cite{MR1838752,klindex}]\label{thm:cindex}
	If the category of superselection sectors comes from a rational conformal net as in Example~\ref{ex:rcft} we have $\mc{D}^2 = \mu_{\pi_0}$. In the case of Example~\ref{ex:topo}, we have $\mathcal{D}^2 \leq \mu_{\pi_0}$. Here $\mathcal{D}^2$ is the total quantum dimension of the tensor category of the sectors of the theory, as before.
\end{theorem}
Even though the statement for the cone-localized charges in quantum spin systems is somewhat weaker, in practice it turns out that the bound is actually saturated and we have an equality. For example one can consider the class of models defined by Kitaev~\cite{MR1951039} for a finite abelian group $G$. This can be considered in the setting of Example~\ref{ex:topo}, where $\pi_0$ is the GNS representation of the translation invariant ground state (which is unique). The index can be calculated explicitly in these models, and is equal to $|G|^2$, which is equal to $\mathcal{D}^2$ for these models~\cite{phdleander,klindex}. 

Let us consider the case $G = \mathbb{Z}_2$ in more detail~\cite{toricendo}. Recall that it is defined on the edges of $\mathbb{Z}^2$. An important role in the theory is played by \emph{path} or \emph{string} operators. To a finite path $\xi$ of edges we can associate an operator $F_\xi$, by acting with $\sigma_z$ on each edge. Similarly, we can take a path on the dual lattice $\widehat{\xi}$ and identify it with all the edges it crosses. To such a dual path we assign the path operator $F_{\widehat{\xi}}$ by acting with $\sigma_x$ on the edges of the dual path. Note that both operators are self-adjoint and square to the identity (since the Pauli matrices do). These operators play a fundamental role because they create excitations: if $\Omega$ is the ground state vector, $F_\xi \Omega$ is a state with two excitations, which can be thought of as being located at the end of the string. A similar statement is true for $F_{\widehat{\xi}}$. To get examples of representations satisfying~\eqref{eq:sselect}, we can choose a semi-infinite string $\xi$. If $\xi(n)$ is the finite path consisting of the first $n$ edges, one can show that $\rho(A) = \lim_{n \to \infty} F_{\xi(n)} A F_{\xi(n)}$ defined an automorphism of $\alg{A}$, and $\pi_0 \circ \rho$ satisfies equation~\eqref{eq:sselect}. The charge transporters can be constructed explicitly as well.

With notation as before, it can be shown that $\widehat{\mc{R}}_{AB}$ is generated by $\mc{R}_{AB}$ and two charge transporters $V_X$ and $V_Z$, which both square to the identity and can be chosen to commute. We set $V_0 := I$ and $V_Y := V_X V_Z$, so that we get a unitary representation $g \mapsto V_g$ of $\mathbb{Z}_2 \times \mathbb{Z}_2$ in a natural way. With this notation, it can be shown that every element $X \in \widehat{\mc{R}}_{AB}$ can be uniquely written as $X = \sum_i A_i V_i$ with $A_i \in \mc{R}_{AB}$. The conditional expectation $\mc{E} : \widehat{\mc{R}}_{AB} \to \mc{R}_{AB}$ is then given by $\mc{E}\left(\sum_i A_i V_i\right) = A_0$. As mentioned, the index is equal to $|G|^2 = 4$~\cite{klindex}.

\section{Finite dimensional approximation}\label{sec:finapprox}
Consider again the setting of the toric code in the thermodynamic limit. We will continue to use the notation of Section~\ref{sec:examples}. Recall that we are interested in relative entropies $S(\omega, \varphi)$, where $\omega$ and $\varphi$ are both normal states of $\mc{R}_{AB}$ or $\widehat{\mc{R}}_{AB}$. Unfortunately, in this setting the definition of the relative entropy is rather technical, making it difficult to do concrete calculations. In addition, for the toric code and similar systems on finite lattices we have a much better understanding of, for example, the entanglement entropy in ground states. Hence it is desirable to be able to at least find an approximation of $S(\omega, \varphi)$ in terms of finite dimensional systems.

\begin{figure}
	\begin{center}
	\begin{tikzpicture}
		\clip (-5.1,-1) rectangle (6,4.9);
		\draw[draw=black] (-5,6) -- (0,-0.5) -- (-1.5,6) -- (-3,6);
		\draw[draw=black] (3,6) -- (1.5,-1) -- (6.5,6) -- (3,6);
		\draw (-3.45,3.5) node {$C_1$};
		\draw (4.4,2.5) node {$C_2$};

		\begin{scope}
			\clip (-5,6) -- (0,-0.5) -- (-1.5,6) -- (-3,6);
			\draw[fill=lightgray,thin] (0,-0.5) circle [radius=4];
		\end{scope}
		\begin{scope}
			\clip (3,6) -- (1.5,-1) -- (6.5,6) -- (3,6);
			\draw[fill=lightgray,thin] (1.5,-1) circle [radius=4];
		\end{scope}
	
		\draw[very thick] (-0.95,1.15) .. controls (-2.5,3) and (-0.5,5) .. (-4.5,6);
		\draw[very thick] (2.1,1) .. controls (4.5,3.2) and (3.5,4.75) .. (5.5,6);
		\draw (-0.80,1.95) node {$\xi_{X}^1(n)$};
		\draw (2.40,2) node {$\xi_X^2(n)$};
	\end{tikzpicture}
	\end{center}
	\caption{Two cones $C_1$ and $C_2$. The shaded part is an example of all sites in $\mc{R}_n$ for some $n$. Also shown are semi-infinite strings used in the definition of $V_X$. The paths $\xi_X^1(n)$ and $\xi_X^2(n)$ are the finite parts of these paths obtained by taking the intersection with the shaded parts.}
\label{fig:leftcone}
\end{figure}
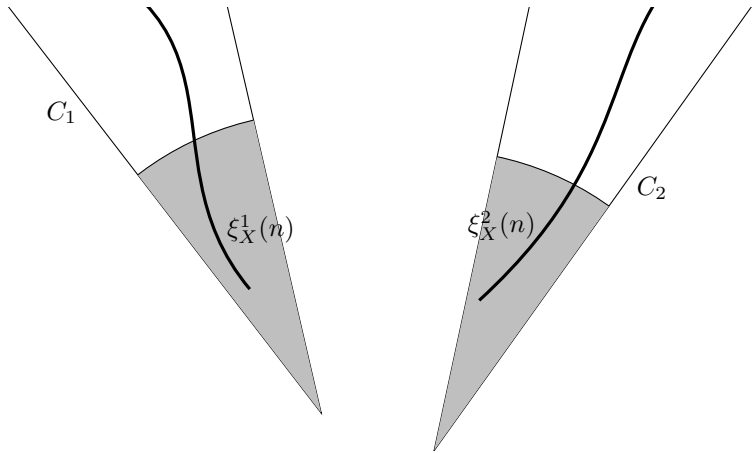

We show how this can be done in our standing example of the toric code. Consider two cones $C_k$ as before, and write $x_k$ for the tips of the cones. Define $\Lambda_n = \bigcup_{k=1,2} C_k \cap B_{n}(x_k)$. That is, the set of all spins in one of the cones with distance at most $i$ from the tip. We define the algebra $\mathcal{R}_n := \pi_0(\alg{A}(\Lambda_n))$.

We also choose two semi-infinite paths to infinity, one in each cone. We write $\xi_X^{i}(n)$ for the intersection with $\Lambda_n$ (see also Figure~\ref{fig:leftcone}). Then we can define the charge transporter $V_X$ as the weak-operator limit of the sequence of path operators $F_{\xi_X^1(n)} F_{\xi_X^2(n)} F_{\xi_n'}$, where $\xi_n'$ is a path connecting the far ends of $\xi_X^1(n)$ and $\xi_X^2(n)$ using the shortest path completely outside of the shaded region $\Lambda_n$ (see~\cite{toricendo} for a proof of convergence and that this indeed is a charge transporter). Similarly, we choose two semi-infinite dual paths $\xi_Z^i$, one in each cone. We choose them in such a way that the shortest (dual) path between any two points (with one on each path) never crosses $\xi_X^1$ or $\xi_X^2$. For example, in Figure~\ref{fig:leftcone} we could choose a path to the right of $\xi_X^1$, and one to the left of $\xi_X^2$. We can then define $V_Z$ analogously to the definition of $V_X$. There are some properties of this construction that we will need later: with the choices we made it follows that $V_X V_Z = V_Z V_X$. We also have that $V_X^2 = V_Z^2 = I$. Finally, note that $F_{\xi_X^1(n)} F_{\xi_X^2(n)} V_X$ commutes with \emph{all} operators in $\mc{R}_n$, since the two string operators effectively cancel the action of $V_X$ in the shaded region. A similar statement is true for $V_Z$.

With this notation we then define $\widehat{\mathcal{R}}_n := \pi_0(\alg{A}(\Lambda_n)) \vee \{V_X, V_Z\}$. Together with $\mc{R}_n$ these two algebras will be used to approximate the index for $\mc{R}_{AB} \subset \widehat{\mc{R}}_{AB}$, using the following lemma.
\begin{lemma}\label{lem:vnapprox}
	The sequence $\mathcal{R}_i$ is an increasing net of finite dimensional von Neumann algebras such that $\left( \bigcup_i \mathcal{R}_i \right)'' = \mathcal{R}_{AB}$, and similarly the weak$^*$ closure of the union of $\widehat{\mc{R}}_i$ is equal to $\widehat{\mc{R}}_{AB}$. Moreover, $\mathcal{E}$ restricts to a conditional expectation $\mathcal{E}_i : \widehat{\mc{R}}_i \to \mc{R}_i$.
\end{lemma}
\begin{proof}
	Note that $\mc{R}_i$ is by definition the tensor product of finitely many spin-$\frac{1}{2}$ algebras, and hence finite dimensional. By definition $\bigcup_i \mc{R}_i$ is dense in $\pi(\alg{A}(C_1 \cup C_2))$ in the norm topology. This algebra in turn is by definition weak$^*$-dense in $\mc{R}_{AB}$.

	Recall that $\widehat{\mc{R}}_{AB}$ is equal to $\mc{R}_{AB} \vee \{V_X, V_Z\}$ (by~\cite[Lemma 4.1]{klindex}), from which the claim that $\bigcup \widehat{\mc{R}}_i$ is weak$^*$ dense in $\widehat{\mc{R}}_{AB}$ follows. In addition, since $V_X$ and $V_Z$ square to the identity and mutually commute, the algebra that they generate is finite dimensional. Also note that if $X$ is a product of path operators, we have $X V_X = \pm V_X X$, and similarly for $V_Z$ (see the proof of~\cite[Lemma 4.4]{klindex}). It follows that $\widehat{\mc{R}}_i$ is finite dimensional.

	The claim on the conditional expectations is clear from the definition of $\mc{E}$.
\end{proof}

The key point is that Lemma~\ref{lem:vnapprox} gives us a way, through the relation of the index to an entropic quantity, to obtain an approximation in terms of finite dimensional systems. More precisely, under these conditions we have that the sequence $S(\omega|\mc{R}_i, \varphi|\mc{R}_i)$ converges to $S(\omega, \varphi)$ by Corollary 5.12(iv) of~\cite{MR1230389}.

The advantage of using a finite dimensional approximation is that it is possible to use density matrix techniques instead of spatial derivatives. In general this makes the problem much more amenable to direct calculation. To this end it will be useful to have a more explicit description of $\widehat{\mc{R}}_i$. Since it is a finite dimensional $C^*$-algebra, it is isomorphic to an algebra of the form $\bigoplus_k M_{n_k}(\mathbb{C})$ for some integers $n_k$. The next lemma gives this decomposition. As we will see later, this decomposition is related to the superselection sectors of the model.

\begin{lemma}\label{lem:findecomp}
	For each $i$ there is an integer $n_i$ such that $\mathcal{R}_i \cong M_{n_i}(\mathbb{C})$ and $\widehat{\mathcal{R}}_i \cong \bigoplus_{j=1}^4 M_{n_i}(\mathbb{C})$. In this representation the algebra $\mathcal{R}_i$ is embedded into $\widehat{\mathcal{R}}_i$ via $A \mapsto \diag(A,A,A,A)$. Together with $\widehat{V}_X = \diag(I, -I, I, -I)$ and $\widehat{V}_Z = \diag(I, I, -I, -I)$ they generate the whole algebra.
\end{lemma}
\begin{proof}
	The algebra $\mc{R}_i$ is a tensor product of finitely many matrix algebras $M_2(\mathbb{C})$, and hence is isomorphic to $M_{n_i}(\mathbb{C})$ for some $n_i$. As for $\widehat{\mc{R}}_i$, first write $F_X := F_{\xi_X^1(i)} F_{\xi_X^2(i)}$, and similarly for $F_Z$. Note that $F_X V_X$ and $F_Z V_Z$, together with $\mc{R}_i$, still generate $\widehat{\mc{R}}_i$. But using the discussion before Lemma~\ref{lem:vnapprox}, $\frac{1}{2}(I \pm F_X V_X)$ and $\frac{1}{2}(I \pm F_Z V_Z)$ are easily checked to be projections with mutually orthogonal ranges in the center of $\widehat{\mathcal{R}}_i$. Hence we have four central projections that we can use to decompose $\widehat{\mc{R}}_i$ into four blocks.

This can be done explicitly: setting $\pi(A) = \diag(A,A,A,A)$ and $\pi(F_X V_X) = \widehat{V}_X$ for $k = X,Z$ we obtain an injective map from $\widehat{\mc{R}}_i$ into $\bigoplus_{i=1}^4 M_n(\mathbb{C})$. It is straightforward to check that it preserves the algebraic relations of $\mc{R}_i$. We also note that $\frac{1}{4}\pi( (I+F_X V_X)(I+F_Z V_Z) ) = \diag(I,0,0,0)$, that is, the projection onto the first block. With similar expressions we obtain the projections onto the other blocks, and we see that $\pi$ is also surjective. This completes the proof.
\end{proof}

\begin{remark}
	Note that the lemma in particular shows that $\widehat{\mc{R}}_i$ is \emph{not} a factor. The reason is that we can conjugate the charge transporters $V_k$ with a unitary in $\mc{R}_i$, such that this conjugated operator commutes with all operators in $\widehat{\mc{R}}_i$, hence the center is non-trivial. This is no longer true for the infinite algebra $\widehat{\mc{R}}_{AB}$. Hence $\mc{R}_i \subset \widehat{\mc{R}}_i$ is not a subfactor, and even though we have a conditional expectation between them, we therefore cannot apply the theory of subfactors directly. In addition, for finite Type I factors the theory behaves a bit differently (see for example~\cite[Thm 3.8]{MR1662525}). 
\end{remark}

We now come to the definition of a state $\omega = \omega \circ \mc{E}$ that will be used to approximate the index. First define projections $P_{\pm} := \frac{1}{2}(I \pm V_X)$ and $Q_{\pm} := \frac{1}{2}(I \pm V_Z)$. Note that they all mutually commute. Moreover, $P_+ P_- = Q_+ Q_- = 0$.

It will be useful to calculate expectation values $\omega_0(P_j Q_k A P_j Q_k)$ with $A \in \mathcal{R}_i$. To this end, first note that $\omega_0(V_j A V_k) = 0$ if $j \neq k$. Heuristically this can be understood in terms of superselection sectors: the operators $V_j$ create a pair of excitations, one in each of the cones $C_i$. Hence they can interpolate between the different sectors in each cone. The local operators $A$, on the other hand, can only create pairs of excitations in each cone individually, i.e.\ after acting with $A$ you stay in the same sector. It follows that $\langle V_j^* \Omega, A V_k \Omega \rangle = 0$, where $\Omega$ is the GNS vector of $\omega_0$. A rigorous argument can be distilled from the results in~\cite{toricgs}. With this observation, we find
\begin{equation}
	\label{eq:projected}
	\begin{split}
	\omega_0(P_i Q_j A &V_k P_i Q_j) = \\
			&\frac{c(i,j,k)}{16}\left( \omega_0(A) + \omega_0(V_X A V_X) + \omega_0(V_Y A V_Y) + \omega_0(V_Z A V_Z) \right), 
	\end{split}
\end{equation}
for $A \in \alg{A}(C_1 \cup C_2)$. Here $c(i,j,k)$ is given by
\[
	c(i,j,0) = 1, \quad c(i,j,X) = i, \quad c(i,j,Y) = i \cdot j, \quad c(i,j,Z) = j,
\]
where $i,j \in \{+,-\}$  and we identify $\pm$ with $\pm 1$.

We can now define a state $\omega$ on $\widehat{\mc{R}}_{AB}$ that will be used to estimate the relative entropy of $\mc{R}_{AB}$ and $\widehat{\mc{R}}_{AB}$:
\begin{equation}
	\label{eq:omega}
	\omega(X) := \sum_{j,k \in \{+, -\}} \omega_0(P_j Q_k X P_j Q_k).
\end{equation}
From the previous calculation it follows that $\omega(I) = 1$, and hence it is a state. Moreover, by collecting terms one sees that $\omega(A V_k) = 0$ if $k \neq 0$, and hence $\omega \circ \mathcal{E} = \omega$. Importantly, this is \emph{not} true for the individual states in the decomposition~\eqref{eq:omega}. 

\begin{lemma}
	The state $\omega$ defined in equation is a normal state on $\widehat{\mc{R}}_{AB}$ such that $\omega \circ \mc{E} = \omega$. It can be written as an equal weight superposition of four \emph{distinct} states, but whose restrictions to $\mc{R}_{AB}$ are equal. The same is true when we restrict $\omega$ to $\widehat{\mc{R}}_i$.
\end{lemma}
\begin{proof}
	Since $\omega_0$ is implemented as a vector state for $\widehat{\mc{R}}_{AB}$ (the GNS vector for $\pi_0$), normality is clear. Using equation~\eqref{eq:projected} we see that the states are distinct, and that the observables $V_k$ can be used to distinguish them. However, the same equation shows that when restricted to $\mc{R}_{AB}$, all states have the same expectation values. The last claim follows directly from the construction.
\end{proof}

By Lemma~\ref{lem:findecomp}, the algebra $\widehat{\mc{R}}_i$ is a direct sum of four copies of $M_n(\mathbb{C})$. To calculate the relative entropies that we need, it will be useful to explicitly find the density operators in this representation. First define $\omega_{jk}(X) := 4 \omega_0(P_j Q_k X Q_k P_j)$ with $j,k = \pm$. The factor of four ensures that $\omega_{jk}(I) = 1$, so that it is in fact a state. The corresponding density matrices can then be written as:
\begin{lemma}
\label{lem:densityop}
The density matrix for $\omega_{++}$ is given by
\[
	\rho_{++} = \frac{1}{16} \begin{bmatrix}
		F_{++} \rho_0 F_{++} & & & \\
		& F_{-+} \rho_0 F_{-+} & &  \\
		& & F_{+-} \rho_0 F_{+-} &   \\
		& & & F_{--} \rho_0 F_{--} \\
	\end{bmatrix}.
\]
Here $\rho_0$ is the density matrix of $\omega_0$ restricted to $\mc{R}_i$ and $F_{\pm\pm} = (I \pm F_X)(I \pm F_Z)$ (the first $\pm$ in the subscript corresponds to the first $\pm$ in the product, and similarly for the second). The same is true for the other combinations of $P_\pm$ and $Q_\pm$ by flipping the appropriate signs.
\end{lemma}
\begin{proof}
	Let $\rho_0$ be as in the statement. Then $\rho := \frac{1}{4} (\rho_0 \oplus \rho_0 \oplus \rho_0 \oplus \rho_0)$ is the density matrix of $\omega_0$ restricted to $\widehat{\mc{R}}_i$ in the representation of Lemma~\ref{lem:findecomp}: it is easy to check that $\Tr(\rho \diag(A,A,A,A)) = \omega_0(A)$ for all $A \in \mc{R}_i$, while $\Tr(\rho \diag(A,A,A,A) \widehat{V}_k) = 0$ for $k \neq 0$. Hence $4 P_+ Q_+ \rho P_+ Q_+$ is the density operator for $\omega_{++}$. By writing again $V_X = F_X \widehat{V}_X$ as in the proof of Lemma~\ref{lem:findecomp} we can write $P_+$ and $Q_+$ in the representation of Lemma~\ref{lem:findecomp}, from which we obtain $\rho_{++}$.

	Alternatively the correctness can be be verified directly using equation~\eqref{eq:projected}. The other choices of $\pm$ are shown in the same way.
\end{proof}

We are now in a position to approximate $[\widehat{\mc{R}}_{AB} : \mc{R}_{AB}]$ via a limiting procedure. The result below coincides with the value of the index $[\widehat{\mc{R}}_{AB} : \mc{R}_{AB}]$ which has been obtained before using different methods~\cite{klindex}. The main interest is in providing a specific state that realizes the equality, as well as in the proof method of taking the limit of finite dimensional systems. This concrete procedure provides an insight into the physical mechanisms behind the equality.
\begin{theorem}
Let $\mc{E} : \widehat{\mc{R}}_{AB} \to \mc{R}_{AB}$ be the conditional expectation given above. Then
\[
	\sup_{\varphi : \varphi \circ \mc{E} = \varphi} \sup_{ \{p_x\}, \{\varphi_x\}} S_{\widehat{\mc{R}}_{AB}|\mc{R}_{AB}}(\{p_x\}, \{\varphi_x\}) = 2 \log 2 = \log \mathcal{D}^2.
\]
The supremum is over all normal states leaving $\mc{E}$ invariant. It is attained on the state $\omega$ defined in equation~\eqref{eq:omega}, decomposed as an equal-weight convex combination of all four possibilities for $\omega_{\pm\pm}$.
\end{theorem}
\begin{proof}
	For the toric code there are precisely four abelian sectors, so $\mc{D}^2 = \sum_{i=1}^4 d_i^2 = 4$~\cite{toricendo,klindex}, from which the second equality follows. To show the first equality, we will calculate $S_{\widehat{\mc{R}}_{AB}|\mc{R}_{AB}}( \{p_x\}, \{ \varphi_x \} )$ for suitable decompositions. We do this by first restricting all states to $\widehat{\mc{R}}_i$ (and $\mc{R}_i$, respectively) and using the limiting procedure of~\cite[Corollary 5.12(iv)]{MR1230389}.

We first do this for $\varphi = \omega$ defined earlier, to show that the value $2 \log 2$ can be attained. Set $p_x = 1/4$ for $x = 1, \ldots, 4$ and choose each corresponding $\omega_i$ to be a distinct $\omega_{jk}$, $j,k = \pm$, as above. We have already observed that $\omega(A) = \omega_{jk}(A)$ for all $A \in \mc{R}_i$ and $j,k \in \{+, -\}$, hence $S(\omega_{x} | \mc{R}_i, \omega | \mc{R}_i) = 0$. To calculate $S(\omega_x, \omega)$, write $\rho_x$ and $\rho$ for the corresponding density matrices in the representation of Lemma~\ref{lem:findecomp}. Since $(I+F_X)(I-F_X) = (I+F_Z)(I-F_Z) = 0$, Lemma~\ref{lem:densityop} implies that the density matrices $\rho_x$ have mutually disjoint support projections. Because $\rho = \frac{1}{4} \sum \rho_x$ by definition, it follows that 
\[
	S(\rho_x, \rho) = S\left(\rho_x, \frac{1}{4} \rho_x\right) = - \log \left(\frac{1}{4}\right) + S(\rho_x, \rho_x) = 2 \log 2,
\]
independent of $x$. This gives a lower bound on the supremum.

To complete the proof we will show that $2 \log 2$ is the \emph{maximum} value that equation~\eqref{eq:chidiff} can attain, with the extra condition that $\varphi := \sum_x p_x \varphi_x$ satisfies $\varphi \circ \mc{E} = \varphi$. We first characterize such states. By Lemma~\ref{lem:findecomp} the density operator $\rho_\varphi$ for $\varphi$ can be written in the form $\rho_\varphi^1 \oplus \cdots \oplus \rho_\varphi^4$, with each component a positive operator in $\mc{R}_n$. Since $\varphi \circ \mc{E} = \varphi$, we must have that $\varphi(A V_i) = 0$ for $i \neq 0$ and $A \in \mc{R}_i$. But this implies that $\Tr(\rho_\varphi \diag(A,A,A,A) \widehat{V}_X) = 0$ for all $A \in \mc{R}_i$ and $\widehat{V}_X$ as in Lemma~\ref{lem:findecomp}. This can only be true if
\[
	\rho_\varphi^1 - \rho_\varphi^2 + \rho_\varphi^3 - \rho_\varphi^4 = 0.
\]
A similar argument for $\widehat{V}_Z$ and $\widehat{V}_Y := \widehat{V}_X \widehat{V}_Z$ gives a system of linear equations, which has as only solution $\rho_\varphi^i = \rho_\varphi^1$ for $i = 1,\ldots,4$.

Similarly as for the state $\varphi$, the density operators for the states $\varphi_x$ can be decomposed as $\rho_{\varphi_x} =: \rho_x^1 \oplus \cdots \oplus \rho_x^4$. With this notation we find
\[
	S(\varphi_x | \mc{R}_i, \varphi | \mc{R}_i) = \sum_{i=1}^4 \Tr_0 \left[ \rho_x^i \left( \log\left(\sum_i \rho^i_x\right) - \log\left( 4 \rho_\varphi^1\right)\right)\right],
\]
where $\Tr_0$ is the canonical trace on $\mc{R}_i$. Note that $\log (4 \rho_\varphi^1) = 2 \log 2 + \log(4 \rho_\varphi^1)$. Again using that the density operators for $\varphi_x$ and $\varphi$ are block-diagonal, this leads to
\begin{align*}
 S(\varphi_x, \varphi) &- S(\varphi_x | \mc{R}_i, \varphi | \mc{R}_i) \\
	&= \sum_{i=1}^4 \Tr_0 \left[ \rho_x^i \left( \log\left(\rho^i_x\right) - \log\left( \sum_i \rho_x^i\right) + 2 \log 2\right)\right] \\
	&= 2 \log 2 + \sum_i S_0\left(\rho_x^i, \sum_i \rho_x^i\right).
\end{align*}
To conclude the argument it is enough to show that the summands are always less than or equal to zero (note that $\rho^i_x$ is not a state since it is not normalized, so this is not a contradiction with the relative entropy of two states always being positive). Clearly $\rho_x^i \leq \sum_k \rho_x^k$, since each term is a positive operator. We may assume that $\rho_x^i > 0$ since otherwise $S_0(\rho_x^i, \sum_k \rho_x^k) = 0$. Because $\log$ is operator monotone it follows that $\log(\rho_x^i) \leq \log\left(\sum_k \rho_x^k\right)$, so that
\[
	S_0\left(\rho_x^i, \sum_k \rho_x^k\right) = \Tr_0 \left[ \rho_x^i \left(\log \left(\rho_x^i\right) - \log \left(\sum_k \rho_x^k\right) \right) \right] \leq 0
\]
Taking the sum over $p_x$ gives the desired result. The proof is complete by noting that Theorem~\ref{thm:cindex} gives an upper bound to $\log \mc{D}^2$, but from the remark at the beginning of the proof and the calculation here, the bound is an equality.
\end{proof}
\begin{remark}
	In the proof we directly calculated the relative entropies, since they have a clear physical meaning in quantum information theory. Since $\mc{R}_i \subset \widehat{\mc{R}}_i$ is an inclusion of $C^*$-algebras, and we have a faithful conditional expectation $\mc{E}_i$, it is interesting to compare this with Watatani's index theory for $C^*$-algebras~\cite{MR996807}. Note that the inclusion maps $M_n(\mathbb{C})$ onto four copies of itself. Moreover, the standard trace on $\widehat{\mc{R}}_i$ satisfies $\Tr(\mathcal{E}(A)B) = \Tr(AB)$ for all $A \in \widehat{\mathcal{R}}_i$ and $B \in \mathcal{R}_i$. It follows from the analysis in Section 2.4 of~\cite{MR996807} that $\operatorname{Index} \mathcal{E}_i = 4 (I \oplus I \oplus I \oplus I) \in \widehat{\mc{R}}_i$.
\end{remark}

For quantum double models for general finite abelian groups $G$ (the toric code corresponds to $G = \mathbb{Z}_2$) the structure is very similar~\cite{phdleander,haagdouble}, and with the appropriate modifications the same proof goes through with $\mc{D}^2 = |G|^2$, although carrying it out explicitly is much more involved. For non-abelian $G$ we expect the structure to be a bit different: there the irreducible representations of the quantum double $\mc{D}(G)$, the symmetry algebra behind the model, are no longer one dimensional. These irreducible representations are in correspondence with the different charges of the model. As a consequence, we expect that $\widehat{\mc{R}}_i$ no longer decomposes into blocks of equal size, but rather dependent on the dimension of the irreducible representations of $\mc{D}(G)$.

We conclude with a discussion of the properties that made the finite dimensional approximation work. First of all, Lemma~\ref{lem:findecomp} is a direct manifestation of superselection sectors, in the sense that local operators cannot interpolate between the different sectors. Such a structure is expected for all topologically ordered models. To do the explicit computation here, it was helpful that the sequence of approximating algebras is essentially the same for each $n$. This is due to the convenient choice of $V_X$ and $V_Z$, in such a way that we can easily cancel their action inside the shaded regions in Figure~\ref{fig:leftcone}. It seems reasonable to conjecture that as long as excitations are created by string-like operators, a similar structure is true. Finally, one needs to find a suitable state $\omega = \omega \circ \mc{E}$. In particular, it is useful if $\mc{E}$ easily restricts to the finite dimensional algebras. Indeed, this was necessary to characterize all invariant states in the theorem.

We also note that the structure compares closely to a construction due to Haah~\cite{MR3465431}. He characterizes operators that do not change the charge within an annulus, while here we consider the dual viewpoint: we take those operators that create a pair of excitations, one in each annulus. Indeed, the operators $V_X, V_Y$ and $V_Z$ play that role. In the finite setting we can replace these by suitable string operators. These are precisely characterized by the condition that they do not create any excitations outside of the two parts of the cone (the path outside the cone is not important, and any choice will generate the same algebra $\widehat{\mc{R}}_i$). This is reminiscent of Haah's condition that certain operators should commute with the projections in the Hamiltonians he considers. As a matter of fact, in the end Haah obtains a similar algebraic structure as we do (in particular, equation~(3) of~\cite{MR3465431}).

\bibliographystyle{amsplain}
\bibliography{refs/refs.bib}

\end{document}